\documentclass[runningheads,envcountsame]{llncs}
\usepackage{tikz}
\usetikzlibrary{arrows,shapes,trees}
\usepackage{latexsym}
\usepackage{amssymb}  
\usepackage{amsmath}
\usepackage{stmaryrd}

\begin{document}
 \title{
 Model Checking Markov Chains Against Unambiguous B\"{u}chi
   Automata
 } 
 \author{Michael Benedikt, Rastislav Lenhardt, and James Worrell}
 \institute{Department of Computer Science, University of Oxford, UK }
 \maketitle 
\section{Erratum}
The authors would like to withdraw the claimed proof of
Theorem~\ref{thm:main} in the note below.  The problems with the proof
(and the overall approach, based on the notion of recurrent states)
are detailed in a counterexample that can be found at
\texttt{https://s3-eu-west-1.amazonaws.com/cav16-uba/counterexample.pdf}.
In particular, the above document gives an example of a Markov chain
$M$ and automaton $A$ such that $P_M(L(A))$ is strictly positive but
for which the product $M\otimes A$ has no recurrent state.  This shows
that Equation (\ref{eq:lang2}) does not hold.  The problem here stems
from Equation~(\ref{eq:mistake}), which is invalid.

\section{Introduction}
An automaton is \emph{unambiguous} if each word has at most one
accepting run and \emph{separated} if no word is accepted from two
distinct states.  The classical translation of LTL formulas to
B\"{u}chi automata~\cite{VardiW86} produces unambiguous separated
automata since the states of such automata correspond to complete
subformula types.  Motivated by this observation,
Couvreur \emph{et al.}~\cite{CouvreurSS03} present a
polynomial-time algorithm to model check Markov chains against
separated unambiguous B\"{u}chi automata

In this note we give a polynomial-time algorithm for model checking
Markov chains against B\"{u}chi automata that are unambiguous but not
necessarily separated.  Apart from the extra generality of this
procedure, our main motivation is the fact that the
\emph{build-by-need} translation from LTL to B\"{u}chi automata
described in~\cite{BenediktLW13b}---adapting the construction
of~\cite{GerthPVW95}---produces automata that are unambiguous but
which may not be separated.

\section{Definitions}
We briefly recall the main definitions.  See~\cite{CY95,CouvreurSS03}
for more details.

A \emph{Markov chain} $M=(S,P,\pi)$ consists of a set $S$ of
\emph{states}, a \emph{transition probability function} $P:S\times S
\rightarrow [0,1]$ such that $\sum_{t\in S} P(s,t)=1$ for each state
$s\in S$, and an \emph{initial probability distribution} $\pi$ on $S$.
We assume that all numerical data are rational.  

We denote by $\Pr_M(L)$ the probability that $M$ performs a trajectory
in a given measurable set $L\subseteq S^\omega$.  We extend this
notation to sets of finite words $L\subseteq S^*$, writing $P_M(L)$ as
shorthand for $P_M(LS^\omega)$.

A \emph{non-deterministic automaton} $A=(\Sigma,Q,Q_0,\delta,F)$
comprises a finite \emph{alphabet} $\Sigma$, a finite set of
\emph{states} $Q$, set of \emph{initial states} $Q_0 \subseteq Q$,
\emph{transition function} $\delta:Q\times \Sigma \rightarrow 2^Q$,
and set of \emph{accepting states} $F$.  We extend $\delta$ to a
function $\delta:Q\times\Sigma^+ \rightarrow 2^Q$ by inductively
defining $\delta(q,w\sigma) = \bigcup \{\delta(q',\sigma) : q'\in
\delta(q,w)\}$, where $w\in \Sigma^+$.  We consider automata
alternatively as acceptors of finite words and acceptors of infinite
words (via the B\"{u}chi acceptance condition).  In the former case we
speak of \emph{non-deterministic finite automata (NFA)} and in the
latter case of \emph{non-deterministic B\"{u}chi automata (NBA)}.  In
either case we write $L(A)$ for the language accepted by $A$.

\section{Main Result}
Let $M=(S,P,\pi)$ be a Markov chain, $A$ an unambiguous NBW with
alphabet $\Sigma$, and $\lambda : S \rightarrow \Sigma$ a function
labelling the states of $M$ with letters from the alphabet of $A$.
Write $||M||$ and $||A||$ for the respective lengths of the
representations of $M$ and $A$, assuming that integers are encoded in
binary.  We show how to compute $\Pr_M\{ s_1s_2\ldots \in S^\omega :
\lambda(s_1)\lambda(s_2)\ldots \in L(A)\}$---the probability that a
trajectory of $M$ is accepted by $A$---in time polynomial in $||M||$
and $||A||$.  

Without loss of generality, by first applying an existential renaming
to $A$ along $\lambda$, we assume that the alphabet of $A$ is the set
of states of $M$, i.e., $\Sigma=S$, and the state-labelling map
$\lambda$ is the identity.  Note that unambiguous automata are
preserved under existential renaming.  Our task is now to compute
$\Pr_M(L(A))$.  We first consider the task of determining whether
$\Pr_M(L(A))>0$.

\begin{lemma}
Let $M=(S,P,\pi)$ be a Markov chain and
$A=(S,Q,Q_0,\delta,F)$ an unambiguous NFA.  Then $\Pr_M(L(A))$ is
computable in time polynomial in $||A||$ and $||M||$.
\label{lem:calculate}
\end{lemma}
\begin{proof}
Let $L(A,q) \subseteq S^*$ denote the set of words accepted by $A$
with $q\in Q$ as initial state.  Similarly let $\Pr_{M,s}$ denote the
probability distribution on $S^\omega$ induced by $M$ with initial
state distribution $P(s,-)$.  Without loss of generality, assume that
$S$ contains a state $s_0$ with $P(s_0,s) = \pi(s)$ for each $s \in
S$.  Let us also assume that every state in $A$ is reachable from
$Q_0$ and can reach $F$.

Define a directed graph $G_{M \otimes A}=(V,E)$, with set of vertices
$V=S\times Q$ and $(s,q)\mathrel{E}(s',q')$ if and only if $P(s,s')>0$
and $q'\in \delta(q,s')$.  Say that a vertex $(s,q) \in V$ is
\emph{accepting} if $q\in F$ and \emph{dead} if it cannot reach an
accepting vertex.  Write $V^{\mathit{acc}}$ and $V^{\mathit{dead}}$
for the respective sets of accepting and dead vertices, and write $V^?
= V \setminus (V^{\mathit{acc}} \cup V^{\mathit{dead}})$.

Introduce a real-valued variable $\xi_{s,q}$ to represent
$\Pr_{M,s}(L(A,q))$, so that $\sum_{q \in Q_0}\xi_{s_0,q}$ represents
$\Pr_M(L(A))$.  We claim that the following system of equations
uniquely defines $\xi_{s,q}$:
\begin{align}
\xi_{s,q} &=  0  & (s,q) \in V^{\mathit{dead}} \label{eq:first}\\
\xi_{s,q} &=  1 &  (s,q) \in V^{\mathit{acc}} \label{eq:second} \\
\xi_{s,q} &=  \sum_{s'\in S} \, \sum_{q' \in \delta(q,s')}
                P(s,s') \cdot \xi_{s',q'} & (s,q)\in V^?
\label{eq:third}
\end{align}

The correctness of (\ref{eq:first}) and (\ref{eq:second}) is
self-evident.  Correctness of (\ref{eq:third}) follows from the
following calculation:
\begin{eqnarray*}
\xi_{s,q}&=& \mathrm{Pr}_{M,s}(L(A,q))\\
        &=& \sum_{s'\in S} P(s,s') \cdot \mathrm{Pr}_{M,s'} \Big[
\bigcup_{q'\in \delta(q,s')} L(A,q')\Big]\\
        &=& \sum_{s'\in S} \sum_{q'\in \delta(q,s')} 
            P(s,s') \cdot  \mathrm{Pr}_{M,s'}(L(A,q'))
           \qquad\mbox{$A$ is unambiguous}\\
       & = & \sum_{s' \in S} \sum_{q' \in \delta(q,s')} 
                P(s,s') \cdot \xi_{s',q'} \, .
\end{eqnarray*}

To see that the solution of (\ref{eq:third}) is unique, write the
equation system in matrix form as $\boldsymbol{\xi} =
C\vec{\xi}+\vec{d}$, where $\boldsymbol{\xi} = \{ \xi_{(s,q)} : (s,q)
\in V^? \}$, 
\[ C_{(s,q),(s',q')} = \left\{ \begin{array}{ll} 
  P(s,s')\; & q' \in \delta(q,s')\\
    0 &  \mbox{otherwise}
\end{array}\right . \;\;\mbox{ and }\;\;
  d_{(s,q)} = \sum_{s': \delta(q,s')\cap F\neq\emptyset} P(s,s') \, .\]

  Given two solutions $\vec{\xi}$ and $\vec{\xi'}$,
we have $\vec{\xi}-\vec{\xi'} =
C^n(\vec{\xi}-\vec{\xi'})$ for all $n$.  We will show that $\lim_n
C^n=0$, which proves uniqueness.

The entry of index $(s,q)$ in $(I+C+\cdots + C^n) \vec{d}$
is $\Pr_{M,s}(L(A,q) \cap S^{\leq n})$, which converges
to  $\Pr_{M,s}(L(A,q))$ as $n$ tends to infinity.
It follows that $\lim_n C^n(I+C+\cdots + C^m)\vec{d}=0$ for any
fixed $m \in \mathbb{N}$.  But, since all vertices in $V^?$ can reach
$V^{\mathit{acc}}$, there exists some $m$ such that 
$(I+C+\cdots + C^m)\vec{d}$ is strictly positive in every entry.  We
conclude that $\lim_n C^n=0$.

Since systems of linear equations can be solved in polynomial time, the result
follows.\qed
\end{proof}

We now use Lemma~\ref{lem:calculate} to handle the case of automata
over infinite words.  In particular we use the lemma to classify
states of the product $M\otimes A$ as \emph{recurrent} or not.
\begin{theorem}
Let $M=(S,P,\pi)$ be a Markov chain and $A=(S,Q,Q_0,\delta,F)$ an
unambiguous NBA.  Then $P_M(L(A))$ is computable in time polynomial in
$||M||$ and $||A||$.
\label{thm:main}
\end{theorem}
\begin{proof}
Given $(s,q)\in S\times F$,
define $G_{s,q},H_{s,q} \subseteq S^+$ by
\begin{eqnarray*}
G_{s,q} &=& \{ s_1\ldots s_k \in S^+ : s_k=s \mbox{ and }
q \in \textstyle\bigcup_{p \in Q_0} \delta(p,s_1\ldots s_k)\} \\
H_{s,q} &=& \{ s_1\ldots s_k \in S^+ : s_k=s \mbox{ and }
q \in \delta(q,s_1\ldots s_k)\} \, .
\end{eqnarray*}
Thus $G_{s,q}$ is the set of finite trajectories of $M$ that end in
state $s$ and which lead $A$ from an initial location to $q$, while
$H_{s,q}$ is the set of finite trajectories of $M$ that end in state
$s$ and that lead $A$ from location $q$ back to itself.

Clearly we can express $L(A)$ as the following $\omega$-regular expression:
\begin{gather}
L(A) = \bigcup_{(s,q) \in S\times F}G_{s,q}H_{s,q}^\omega \, .
\label{eq:lang}
\end{gather}

Define $(s,q) \in S \times F$ to be \emph{recurrent} if
$\Pr_{M,s}(H_{s,q}) = 1$.  We claim that if $(s,q)$ is recurrent then
$\Pr_{M,s}(H_{s,q}^\omega) = 1$, and if $(s,q)$ is not recurrent then
$\Pr_{M,s}(H_{s,q}^\omega) = 0$.  

Suppose first that $(s,q)$ is recurrent.  Consider the set of
trajectories $S^\omega$ under the measure $\Pr_{M,s}$.  Inductively
define a sequence of random variables $\{ h_n\}_{n \in \mathbb{N}}$ on
$S^\omega$ with values in $\mathbb{N} \cup \{\infty\}$ by writing $h_0
= 0$, and
\[h_{n+1} = \left\{ \begin{array}{ll}
\min \{ k : s_{h_n+1} \ldots s_k \in H_{s,q} \} & \mbox{ if $h_n<\infty$}\\
\infty & \mbox{ otherwise}
\end{array}\right . \]
Then $\Pr(h_{n+1} < \infty \mid h_n < \infty) = \Pr_{M,s}(H_{s,q})=1$.
It follows that $\Pr(\bigcap_n h_n<\infty)=1$ and, \emph{a fortiori},
that $\Pr_{M,s}(H_{s,q}^\omega) = 1$.  On the other hand, if $(s,q)$
is not recurrent then
\begin{gather}
\mathrm{Pr}_{M,s}(H_{s,q}^\omega) = \lim_{n<\omega} \mathrm{Pr}_{M,s}(H_{s,q}^n) =
\lim_{n<\omega} (\mathrm{Pr}_{M,s}(H_{s,q}))^n=0 
\label{eq:mistake}
\end{gather}
and the claim is established.

From Equation (\ref{eq:lang}), we
conclude that
\begin{gather}
\mathrm{Pr}_M(L(A)) = \mathrm{Pr}_M\left(
\bigcup_{(s,q)\;\mathrm{recurrent}} G_{s,q} \right) \, .
\label{eq:lang2}
\end{gather}

Now $H_{s,q}$ is the language of an unambiguous NFA.  The automaton
in question is obtained from $A$ by making $q$ the initial state,
adding a new sink state $q_{\mathit{acc}}$, for every transition $p
\stackrel{s}{\longrightarrow} q$ adding a transition $p
\stackrel{s}{\longrightarrow}q_{\mathit{acc}}$, and making
$q_{\mathit{acc}}$ the unique accepting state.  Thus, by
Lemma~\ref{lem:calculate}, we can determine whether $(s,q)$ is
recurrent in time polynomial in $||M||$ and $||A||$.

The language appearing on the right-hand side of (\ref{eq:lang2}) is
likewise expressible by an unambiguous NFA.  Applying
Lemma~\ref{lem:calculate} once again, we can calculate $\Pr_M(L(A))$
in time polynomial in $||A||$ and $||M||$.\qed
\end{proof}

\bibliographystyle{plain}
\bibliography{erratum}

\begin{thebibliography}{1}

\bibitem{BenediktLW13b}
Michael Benedikt, Rastislav Lenhardt, and James Worrell.
\newblock {LTL} model checking of interval markov chains.
\newblock In {\em TACAS}, volume 7795 of {\em Lecture Notes in Computer
  Science}, pages 32--46. Springer, 2013.

\bibitem{CY95}
C.~Courcoubetis and M.~Yannakakis.
\newblock The complexity of probabilistic verification.
\newblock {\em J. ACM}, 42(4):857--907, 1995.

\bibitem{CouvreurSS03}
Jean-Michel Couvreur, Nasser Saheb, and Gr{\'e}goire Sutre.
\newblock An optimal automata approach to {LTL} model checking of probabilistic
  systems.
\newblock In {\em LPAR}, volume 2850 of {\em Lecture Notes in Computer
  Science}, pages 361--375. Springer, 2003.

\bibitem{GerthPVW95}
Rob Gerth, Doron Peled, Moshe~Y. Vardi, and Pierre Wolper.
\newblock Simple on-the-fly automatic verification of linear temporal logic.
\newblock In {\em PSTV}, volume~38 of {\em IFIP Conference Proceedings}, pages
  3--18. Chapman {\&} Hall, 1996.

\bibitem{VardiW86}
Moshe~Y. Vardi and Pierre Wolper.
\newblock An automata-theoretic approach to automatic program verification
  (preliminary report).
\newblock In {\em LICS}, pages 332--344. IEEE Computer Society, 1986.

\end{thebibliography}

\end{document}